\documentclass[a4paper,12pt,justification=centering]{article}
\usepackage{latexsym}
\usepackage{rotating}
\usepackage{amsmath,amssymb,amsxtra,amscd,amsthm}
\usepackage[mathscr]{eucal}
\usepackage{graphicx}
\usepackage[justification=centering]{caption}
\usepackage{subfig}
\usepackage{datetime}
\usepackage{pdfsync}
\usepackage{thumbpdf}
\usepackage{color}
\ifx\pdfoutput\undefined
\usepackage
[dvips,
 verbose=false,
 bookmarks=true,
 colorlinks=true,
 linkcolor=webred,
 filecolor=webbrown,
 citecolor=webgreen,
 pagecolor=webblue,
 urlcolor=webblue,
 pdftitle={},
 pdfauthor={Murad Alim},
 pdfsubject={},
 pdfkeywords={},
 bookmarksopen=false,
 pdfpagemode=None,
 pdfview=FitH,
 pdfstartview=FitH,
 extension=pdf]{hyperref}
\else
\usepackage
[pdftex,
 verbose=false,
 bookmarks=true,
 colorlinks=true,
 linkcolor=webred,
 filecolor=webbrown,
 citecolor=webgreen,
 pagecolor=webblue,
 urlcolor=webblue,
 pdftitle={},
 pdfauthor={Murad Alim},
 pdfsubject={},
 pdfkeywords={},
 bookmarksopen=false,
 pdfpagemode=None,
 pdfview=FitH,
 pdfstartview=FitH,
 extension=pdf]{hyperref}
\fi
\definecolor{webred}{rgb}{.8,0,0}
\definecolor{webbrown}{rgb}{.6,0,0}
\definecolor{webgreen}{rgb}{0,0.5,0}
\definecolor{webdkgreen}{rgb}{0,0.3,0}
\definecolor{webblue}{rgb}{0,0,0.5}

\numberwithin{equation}{section}

\providecommand{\href}[2]{#2}

\allowdisplaybreaks
\usepackage{bbm}
\usepackage{rotating}
\newcommand{\be}{\begin{eqnarray}}
\newcommand{\beq}{\begin{eqnarray}}
\newcommand{\ee}{\end{eqnarray}}

\def\e#1\e{\begin{equation}#1\end{equation}}
\def\ea#1\ea{\begin{align}#1\end{align}}

\theoremstyle{plain}
\newtheorem{thm}{Theorem}[section]
\newtheorem*{thm*}{Theorem}

\newtheorem{prop}[thm]{Proposition}

\theoremstyle{definition}

\begin{document}

\setlength{\parindent}{0cm}
\setlength{\baselineskip}{1.5em}
\title{Airy Equation for the Topological String \\ Partition Function in a Scaling Limit}

\author{Murad Alim$^{1,2}$\footnote{\tt{alim@physics.harvard.edu}}, Shing-Tung Yau$^1$\footnote{\tt yau@math.harvard.edu} and Jie Zhou$^3$\footnote{\tt jzhou@perimeterinstitute.ca}
\\
\small $^1$ Department of Mathematics, Harvard University,\\ \small 1 Oxford Street, Cambridge, MA 02138, USA\\
\small $^2$ Jefferson Physical Laboratory, Harvard University, \\ \small 17 Oxford Street, Cambridge, MA 02138, USA\\
\small $^3$ Perimeter Institute for Theoretical Physics,\\ \small 31 Caroline Street North, Waterloo, ON N2L 2Y5, Canada }

\date{}
\maketitle

\abstract{We use the polynomial formulation of the holomorphic anomaly equations governing perturbative topological string theory to derive the free energies in a scaling limit to all orders in perturbation theory for any Calabi-Yau threefold. The partition function in this limit satisfies an Airy differential equation in a rescaled topological string coupling. One of the two solutions of this equation gives the perturbative expansion and the other solution provides geometric hints of the non-perturbative structure of topological string theory. Both solutions can be expanded naturally around strong coupling. }

\clearpage




\section{Introduction}

The perturbative expansion of string theories is asymptotic \cite{Gross:1988ib,Shenker:1990uf} which raises questions about the non-perturbative completion and definition of these theories. It has been very fruitful to address these questions within topological string theory which can be connected to Chern-Simons theory and matrix models, see the excellent review~\cite{Marino:2012zq} and references therein. The lessons obtained from this study are expected to give general insights into string perturbation theory.

Topological string theory is based on the nonlinear sigma model which maps the string world-sheet into a target Calabi-Yau (CY) threefold.  The topological string partition function is given by a perturbative expansion in the topological string coupling $\lambda$, summing the free energies $\mathcal{F}^{(g)}$ over the world-sheet genera $g$:
\begin{equation}\label{eqdefofZtoppert}
\mathcal{Z}_{top} = \exp \sum_{g=0}^{\infty} \lambda^{2g-2} \mathcal{F}^{(g)}\,.
\end{equation}

The perturbative free energies $\mathcal{F}^{(g)}$ satisfy the holomorphic anomaly equations of Bershadsky, Cecotti, Ooguri and Vafa (BCOV) \cite{Bershadsky:1993ta,Bershadsky:1993cx} which are recursive differential equations. These were solved using Feynman diagrams \cite{Bershadsky:1993cx}, which give $\mathcal{F}^{(g)}$ the form:
\begin{equation}\label{eqcombinatorics}
\mathcal{F}^{(g)}=\sum_{\Gamma\in  \partial \overline{\mathcal{M}}_{g}} {1\over| \mathrm{Aut}(\Gamma)|}{\omega(\Gamma)}+f^{(g)}\,,
\end{equation}
where $\overline{\mathcal{M}}_{g}$ is the Deligne-Mumford compactification of the moduli space of Riemann surfaces of arithmetic genus $g$. The stratification of this moduli space can be captured by the decorated dual graphs $\Gamma$ of the Riemann surfaces. The first summand in the above equation is a summation over the decorated graphs corresponding to the degenerate Riemann surfaces. The weight $\omega(\Gamma)$ is given by the Feynman rules: for example, an edge in the dual graph corresponds to a node in the degeneration and gives a propagator. The contribution to $\mathcal{F}^{(g)}$ from the smooth Riemann surfaces of arithmetic genus $g$ is the holomorphic ambiguity $f^{(g)}$ which is fixed by boundary conditions. The factors $1/|\mathrm{Aut}(\Gamma)|$  in the above sum are universal constants coming from the structure of $\overline{\mathcal{M}}_{g}$ and are independent of the geometric data of the target CY threefold whose information is encoded in the weights $\omega(\Gamma)$ obtained from the Feynman rules.

The Feynman diagrams become intractable at higher genus since their number grows rapidly. It was shown that this simplifies significantly due to a polynomial structure of the higher genus free energies in finitely many generators \cite{Yamaguchi:2004bt,Alim:2007qj}. The polynomial structure is non-trivial and stems from further decomposing the vertices of the Feynman diagrams at higher genus into simple expressions  as well as from a differential ring structure of the generators. This leads to the much milder polynomial growth of the number of terms at each genus. In this paper we use this polynomial structure to determine at all genera the pieces of the free energies coming from certain monomials at every genus $g$ which include contributions from the most degenerate Riemann surfaces in $\partial \overline{\mathcal{M}}_{g}$. 
To select these terms in the full partition function we rescale the polynomial generators as well as the topological string coupling such that only these terms survive a scaling limit. We derive a modified Bessel differential equation in a rescaled topological string coupling whose solution assembles these monomials at all genera. This solution allows us furthermore to expand this piece of the partition function around strong coupling $\lambda\rightarrow \infty$. The second linearly independent solution of the differential equation is non-perturbative at $\lambda=0$ and provides geometric hints of non-perturbative topological strings. 

Furthermore, a change of variables allows us to transform the differential equation into an Airy differential equation. The latter is the hallmark of two dimensional topological gravity studied using matrix models \cite{Witten:1990hr,Kontsevich:1992ti} and appears in many discussions of non-perturbative phenomena \cite{Marino:2012zq} and especially also in the recent definition of non-perturbative topological strings on some non-compact CY manifolds given in Ref.~\cite{Grassi:2014zfa}. It is surprising that the results of the present work are universal for topological strings on any CY geometry which is subject to the holomorphic anomaly equations even if there is no manifest relation to Chern Simons theory or matrix models such as compact CY threefolds. 

The polynomial structure of the higher genus free energies was already used in Ref.~\cite{Yamaguchi:2004bt}  to determine in principle the coefficients of certain monomials in the generators at all genera for the quintic. In a similar context, where the polynomial generators are realized as quasi modular forms for a non-compact CY given by the canonical bundle over $\mathbbm{P}^1\times\mathbbm{P}^1$, this has been addressed in Ref.~\cite{Huang:2009md}. The mirror curve of the latter geometry is related to the to the matrix model description of ABJM theory \cite{Aharony:2008ug} and has been used in Ref.~\cite{Drukker:2010nc} to obtain strong coupling results. In Ref.~\cite{Fuji:2011km} the all genus free energies for the matrix model are summed up in a certain limit in the moduli space and the coefficients are related to the Airy function which is very close in spirit to our work.\footnote{We would like to thank Marcos Marino for pointing out this reference to us and its relation to our results.} 

The structure of this note is as follows. We review the polynomial structure of topological strings in Sec.~\ref{secpolystructure}. In Sec.~\ref{secBessel} we derive a differential equation which determines the partition function in a limit to all genera and discuss its transformation into the Airy equation as well as the strong coupling expansion of its solutions. We conclude with discussions in Section \ref{seccon}.

\section{Polynomial structure of topological strings}\label{secpolystructure}
In this section we review the polynomial formulation \cite{Yamaguchi:2004bt,Alim:2007qj} of the holomorphic anomaly equations \cite{Bershadsky:1993cx}. Let $\mathcal{M}$ be the moduli space of a Calabi-Yau threefold, which can be the moduli space of complexified K\"ahler structures of a CY threefold $Y$ or the moduli space of complex structures of its mirror $X$. 
In the following we will use local complex coordinates $z^i\,, i=1,\dots,n=\textrm{dim}\,\mathcal{M}$ and restrict to a one-dimensional slice of the moduli space by choosing a local coordinate $z=z^*$ such that $C_{***}$ defined below is non-zero and we set $z^i=0\,,i\ne*$. We use $\partial_z:=\frac{\partial}{\partial z},\partial_{\bar{z}}:=\frac{\partial}{\partial \bar{z}}$. The function $e^{-K}$ gives a Hermitian metric on a line bundle $\mathcal{L}\rightarrow \mathcal{M}$ with connection $K_z$ and provides the K\"ahler potential for the Weil-Petersson metric on $\mathcal{M}$, whose components and Levi-Civita connection are given by $
G_{z\bar{z}} := \partial_z \partial_{\bar{z}} K\,,  \Gamma_{zz}^z= G^{z\bar{z}} \partial_z G_{z\bar{z}}\,.
$ The holomorphic Yukawa couplings are:
$C_{zzz} \in \Gamma\left( \mathcal{L}^2 \otimes \mathrm{Sym}^3 T^*\mathcal{M}\right)\,,
$
the curvature is expressed as :
 \begin{equation} 
-R_{z\bar{z}\phantom{z}z}^{\phantom{z\bar{z}}z}=[\partial_{\bar{z}},D_z]^z_{\phantom{z}z}=\bar{\partial}_{\bar{z}} \Gamma^z_{zz}= 2
G_{z\bar{z}}  - C_{zzz} \overline{C}^{zz}_{\bar{z}},
\label{curvature}
 \end{equation}
where $D_z$ denotes the covariant derivative and $\overline{C}_{\bar{z}}^{zz}:= e^{2K} G^{z\bar{z}} G^{z\bar{z}}\overline{C}_{\bar{z}\bar{z}\bar{z}}.$
This data defines a special K\"ahler manifold  \cite{Strominger:1990pd,Bershadsky:1993cx}.

We further introduce the objects $S^{zz},S^z,S$, which are non-holomorphic sections of $\mathcal{L}^{-2}\otimes \mathrm{Sym}^m T\mathcal{M}$ with $m=2,1,0$, respectively, and give local potentials for the non-holomorphic Yukawa couplings:
\begin{equation}\label{defpropagator}
\partial_{\bar{z}} S^{zz}= \overline{C}_{\bar{z}}^{zz}, \qquad
\partial_{\bar{z}} S^z = G_{z\bar{z}} S^{zz}, \qquad
\partial_{\bar{z}} S = G_{z \bar{z}} S^z.
\end{equation}
These are the propagators of the Feynman rules derived for $\mathcal{F}^{(g)}$ in Ref.~\cite{Bershadsky:1993cx}.

The topological string amplitudes at genus $g$ with $n$ insertions $\mathcal{F}^{(g)}_{n}$  are defined in Ref.~\cite{Bershadsky:1993cx} as non-holomorphic sections of the line bundles
$\mathcal{L}^{2-2g}$ over $\mathcal M$. These are only non-vanishing for
$(2g-2+n)>0$. They are related recursively in $n$ by $D_z \mathcal{F}^{(g)}_{n-1}=\mathcal{F}^{(g)}_{n},$ as well as in $g$ by the holomorphic anomaly equation for $g=1$ \cite{Bershadsky:1993ta}
\begin{equation}
\partial_{\bar{z}} \mathcal{F}^{(1)}_z = \frac{1}{2} C_{zzz}
\overline{C}^{zz}_{\bar{z}}+ (1-\frac{\chi}{24})
G_{z \bar{z}}\,, \label{anom2}
\end{equation}
where $\chi$ is the Euler character of the CY threefold $Y$. As well as for $g\geq 2$ \cite{Bershadsky:1993cx}:
\begin{equation}
\partial_{\bar{z}} \mathcal{F}^{(g)} = \frac{1}{2} \overline{C}_{\bar{z}}^{zz} \left(
\sum_{h=1}^{g-1}
D_z\mathcal{F}^{(h)} D_z\mathcal{F}^{(g-h)} +
D_zD_z\mathcal{F}^{(g-1)} \right) \label{anom1}.
\end{equation}

It was shown in Refs.~\cite{Yamaguchi:2004bt,Alim:2007qj} that for any CY threefold $\mathcal{F}^{(g)}_{n}$ is a polynomial of degree $3g-3+n$ in the generators $S^{zz},S^z,S,K_z$ where degrees $1,2,3,1$ were assigned to these generators respectively. The purely holomorphic part of the construction as well as the coefficients of the monomials are rational functions in the algebraic moduli. 
The recursive proof of this relies on the differential ring structure of the generators. For the purpose of this work we will only need the following \cite{Alim:2007qj}:
\begin{equation}\label{diffring}
D_z S^{zz} = 2  S^z  - C_{zzz}(S^{zz})^2 +
h_z^{zz}\,, 
\end{equation}
where $h_z^{zz}$ denotes a holomorphic function which is fixed by a choice of $S^{zz}$ satisfying Eq.~(\ref{defpropagator}). The expression for the curvature~\eqref{curvature} can be integrated to:
\begin{equation}
\Gamma^z_{zz} = 2 K_z  - C_{zzz} S^{zz} + s^z_{zz}\,,
\label{specgeom}
\end{equation}
with $s_{zz}^z$ a holomorphic function depending on the choice of $S^{zz}$. We use this to write out the following:
\begin{equation}\label{yukpol}
D_{z}C_{zzz} =\partial_z C_{zzz} - 3 s_{zz}^z C_{zzz} - 4 K_z C_{zzz}+ 3  C_{zzz}^2 S^{zz}\,.
\end{equation}
The holomorphic anomaly equations split into two equations \cite{Alim:2007qj}:
\begin{eqnarray}
\frac{\partial \mathcal{F}^{(g)}}{\partial S^{zz}} &=&  \frac{1}{2}
\sum_{h=1}^{g-1}
D_z\mathcal{F}^{(h)} D_z\mathcal{F}^{(g-h)} + \frac{1}{2}
D_z D_z\mathcal{F}^{(g-1)}, \nonumber
\\
0 &=& \frac{\partial \mathcal{F}^{(g)}}{\partial K_z} + S^z \frac{\partial \mathcal{F}^{(g)}}{\partial
S} + S^{zz} \frac{\partial \mathcal{F}^{(g)}}{\partial S^z}\,. \label{polanom}
\end{eqnarray}


\section{All genus differential equation}\label{secBessel}
We determine in the following the all genus coefficients of particular monomials appearing in the polynomial formulation of the free energies $\mathcal{F}^{(g)}$. To this end we derive a differential equation in a rescaled topological string coupling for the partition function which can be transformed into an Airy equation.


\subsection{Scaling limit}

In the polynomial expression of $\mathcal{F}^{(g)}$ we consider the highest degree term in the generator $S^{zz}$ which is a monomial of the form:
$f(z)(S^{zz})^{3g-3} $, where $f(z)$ is a rational function of the modulus $z$. From the Feynman diagram rules of Ref.~\cite{Bershadsky:1993cx} and the polynomial structure reviewed in Sec.~\ref{secpolystructure} we know that it is of the form:
$
f(z)= a_g \,C_{zzz}^{2g-2}, a_{g}\in \mathbb{Q}
$.

The corresponding set of graphs in the Feynman diagrams includes in particular those special graphs $\Gamma$ which correspond to \emph{the most degenerate} Riemann surfaces of arithmetic genus $g$ and of geometric genus $0$. We denote this set by $\overline{\mathcal{M}}_{g,\,\mathrm{cubic}}$. A Riemann surface in this set is obtained by gluing $2g-2$ genus zero Riemann surfaces, each one has three markings, along the markings pairwise. The dual graph is obtained by gluing the cubic vertices along the half-edges. The arithmetic genus is then the number of loops or equivalently the first Betti number of the dual graph.  Among the configurations in $\overline{\mathcal{M}}_{g}$, these dual graphs have the largest possible number of loops. The monomials we focus on furthermore receive contributions from further decomposing the vertices which are given by higher genus amplitudes with insertions.

We introduce the total free energy, omitting $g=0,1$,
\begin{equation}
\mathcal{F} (S^{zz},S^z,S,K_z;z,\lambda) =\sum_{g=2}^{\infty} \lambda^{2g-2} \mathcal{F}^{(g)}(S^{zz},S^z,S,K_z;z)\,.
\end{equation}

To select the terms $a_g \,C_{zzz}^{2g-2}(S^{zz})^{3g-3}$ from $\mathcal{F}$ we rescale the generators $S^{zz}$ as well as the topological string coupling $\lambda$ with $\varepsilon$ in the following way: 
$
\tilde{S}^{zz} = \varepsilon^{2/3} S^{zz}, \tilde{\lambda} = \frac{\lambda}{\varepsilon}
$,
and define:
\begin{equation}\label{scalelimit}
\mathcal{F}_s (\lambda_s) = \lim_{\varepsilon \rightarrow 0} \mathcal{F} (\tilde{S}^{zz},S^z,S,K_z;z,\tilde{\lambda})= \sum_{g=2}^{\infty} a_g \lambda_s^{2g-2} \, \,,
\end{equation}
where we defined the rescaled coupling
$
\lambda_s^2= \lambda^2\, C_{zzz}^{2} (S^{zz})^3
$.

\subsection{Differential equation for all genus free energy}
We will now use the holomorphic anomaly equations in their polynomial form (\ref{polanom}) to derive an equation governing the coefficients $\{a_g\}$ defined in Eq.~(\ref{scalelimit}).

\begin{prop}
The all genus free energy in the scaling limit $\mathcal{F}_s$ satisfies the differential equation:
\begin{equation}\label{free}
\theta_{\lambda_s}^2 \mathcal{F}_s + (\theta_{\lambda_s} \mathcal{F}_s)^2 + 2 \left(1-\frac{2}{3\lambda_s^2} \right) \theta_{\lambda_s} \mathcal{F}_s +\frac{5}{9} =0 \,, \quad \theta_{\lambda_s} = \lambda_s \frac{\partial}{\partial \lambda_s}\,.
\end{equation}
\end{prop}
\begin{proof}

We first introduce the notation:
\begin{equation}
\mathcal{F}^{(g)}= a_g C_{zzz}^{2g-2} (S^{zz})^{3g-3} + \textrm{l.o.t.}\,,\quad g\geq 2\,,
\end{equation}
where we denote by $\textrm{l.o.t.}$ all monomials which vanish in the partition function in the limit described above.
We obtain further from Eqs.~(\ref{diffring}, \ref{yukpol}):
\begin{equation}
D_z \mathcal{F}^{(g)}= (3g-3) a_g C_{zzz}^{2g-1} (S^{zz})^{3g-2} + \textrm{l.o.t.}\,.
\end{equation}
At $g=1$, the integration of Eq.~(\ref{anom2}) becomes:
\begin{equation}\label{genus1}
\mathcal{F}^{(1)}_z:=D_z \mathcal{F}^{(1)}=\frac{1}{2} C_{zzz} S^{zz} + \textrm{l.o.t.}\,.
\end{equation}
At $g=2$, multiplying both sides of Eq.~(\ref{polanom}) by $S^{zz}$ fixes $a_2=\frac{5}{24}$. For $g\ge 3$ the L.H.S. becomes:
\begin{equation}
S^{zz} \frac{\partial \mathcal{F}^{(g)}}{\partial S^{zz}}= a_g (3g-3) C_{zzz}^{2g-2} (S^{zz})^{3g-3} + \textrm{l.o.t.}\,,
\end{equation}
and the R.H.S. is:
\begin{equation}\label{genus3}
\frac{1}{2} \left( \sum_{h=2}^{g-2} a_h a_{g-h} (3h-3)(3g-3h-3) + a_{g-1} (3g-6) (3g-3) \right) C_{zzz}^{2g-2} (S^{zz})^{3g-3} + \textrm{l.o.t.}\,.
\end{equation}
Eq.~(\ref{free}) is obtained from the summation:
\begin{equation}
\sum_{g=2}^{\infty } \lambda^{2g-2} S^{zz} \frac{\partial \mathcal{F}^{(g)}}{\partial S^{zz}}= \sum_{g=2}^{\infty } \lambda^{2g-2} \frac{S^{zz}}{2} \left( \sum_{h=1}^{g-1}D_z\mathcal{F}^{(h)} D_z \mathcal{F}^{(g-h)} + D_z D_z \mathcal{F}^{(g-1)}\right)\,.
\end{equation}
\end{proof}

\subsection{Modified Bessel equation}
The equation for the partition function $\mathcal{Z}_{top,s}=\exp \mathcal{F}_{s}$ in the scaling limit becomes:
\begin{prop} $\mathcal{Z}_{top,s}(\lambda_s)$ satisfies the following differential equation:
\begin{equation} \label{eqmodifiedBessel}
\left( \left(\theta_{\frac{1}{3\lambda_s^2}}\right)^2- \left(\left(\frac{1}{3\lambda_s^2}\right)^2+ \frac{1}{9}\right)\right)  \lambda_s e^{1/3\lambda_s^2}  \mathcal{Z}_{top,s}=0\,, \quad \theta_{\frac{1}{3\lambda_s^2}}:=\frac{1}{3\lambda_s^2}\,  \frac{\partial}{\partial \left(\frac{1}{3\lambda_s^2}\right)} \,.
\end{equation}
\end{prop}
This is the modified Bessel differential equation in terms of the variable $\frac{1}{3\lambda_s^2}$ and the general solution in terms of the modified Bessel functions  $I_{1/3},K_{1/3}$ is given by: 
\begin{equation}\label{eqpartitionfunctionsol}
\mathcal{Z}_{top,s}= \frac{e^{-\frac{1}{3 \lambda_s^2}}}{\lambda_s}  \left(  c_1  I_{\frac{1}{3}}\left(\frac{1}{3
   \lambda_s^2}\right)  + c_2 K_{\frac{1}{3}}\left(\frac{1}{3
   \lambda_s^2}\right)\right) \,.
\end{equation}

Now we discuss the asymptotic behavior of the modified Bessel functions (see e.g., Ref.~\cite{Erdelyi:1981vol2}). If $c_{1}=0$ the series around $\lambda_s=0$ coming from  $K_{1/3}$ is trivial. If $c_{1}\neq 0$, the asymptotic series expansion around $\lambda_s=0$  is independent of $c_{2}$. This gives the sequence $\{a_{g}\}_{g\geq 2}$ up to a constant:
\begin{equation}\label{eqseriesexpansion}
\mathcal{F}_{s}=\left(-{1\over 3\lambda_{s}^2}-\ln \lambda_{s}\right)+\left({1\over 3\lambda_{s}^2}+\ln \lambda_{s}+{5\over 24}\lambda_{s}^{2}+{5\over 16}\lambda_{s}^{4}+
{1105  \over 1152}\lambda_{s}^{6}+\cdots\right)\,.
\end{equation}
The first two terms come from the prefactor in Eq.~(\ref{eqpartitionfunctionsol}) and can be regarded as the contribution of genus zero and genus one free energies to $\mathcal{F}_{s}$ while the rest from the modified Bessel function part.
We fix $c_1=1$ and parameterize the general solution with an additional parameter $\zeta$ by:
\begin{equation}\label{eqnonpertcompletion}
\mathcal{Z}_{top,s}= \frac{e^{-\frac{1}{3 \lambda_s^2}}}{\lambda_s}\left(I_{{1\over 3}}\left({1\over 3\lambda_{s}^2}\right)+\zeta K_{{1\over 3}}\left({1\over 3\lambda_{s}^2}\right)\right)\,,\quad \zeta\in \mathbb{C}\,.
\end{equation}
In particular $\zeta$ does not affect the perturbative expansion around $\lambda_s=0$ but becomes relevant away from this locus. It can be interpreted as giving a non-perturbative correction to the perturbative series of topological string partition function.


\subsection{Strong coupling expansion}
Eq.~(\ref{eqmodifiedBessel}) has two apparent singularities: the one at $\lambda_{s}=0$ is an irregular singularity, while the one at $\lambda_{s}=\infty$ is a regular singularity. 
Expanding $\mathcal{Z}_{top,s}$ near $\lambda_{s}=\infty$, we get
\begin{eqnarray}
\mathcal{Z}_{top,s}
&=&\frac{e^{-\frac{1}{3 \lambda_s^2}}}{\lambda_s} \left(
\left(1-{\pi\over \sqrt{3}}\zeta\right) I_{1/3}\left({1\over 3 \lambda_s^2}\right)+{\pi\over \sqrt{3}}\zeta I_{-1/3}\left({1\over 3 \lambda_s^2}\right)\right)
 \label{eqscalinglimitasymp} \nonumber
\\
&=& 6^{-{1\over 3 }} \left(1-{\pi\over \sqrt{3}}\zeta \right) e^{-\frac{1}{3 \lambda_s^2}}
\lambda_{s}^{-{5\over 3}} \sum_{n=0}^{\infty}  {({1\over 2} {1\over 3\lambda_s^2})^{2n}\over n! \Gamma(n+{4\over 3})}\\ \nonumber
&+&
6^{{1\over 3}}{2\pi\over \sqrt{3}}\zeta  e^{-\frac{1}{3 \lambda_s^2}}   \lambda_s^{-{1\over 3}}    \sum_{n=0}^{\infty}  {({1\over 2} {1\over 3\lambda_s^2})^{2n}\over n! \Gamma(n+{2\over 3})}\,.
\end{eqnarray}
Up to irrelevant factors, one has near $\lambda=\infty$ the following:
\begin{equation}
\mathcal{F}_{s}=-{1\over 3\lambda_{s}^{2}}-{1\over 3}\ln \lambda_{s}+ 6^{-{2\over 3}} {1-{\pi\over \sqrt{3}}\zeta \over {2\pi \over \sqrt{3}} \zeta}\lambda_{s}^{-{4\over 3}}+\mathcal{O}(\lambda_{s}^{-4})\,.
\end{equation}
 
 
\subsection{Airy equation}\label{secAiry}

We can further make the change of variables: $
z=(2\lambda_{s}^2)^{-{2\over 3}}, v=2^{-{1\over 3}}e^{1\over 3\lambda_{s}^{2} } \lambda_{s}^{{1\over 3}} \mathcal{Z}_{top,s} $.
Eq.~(\ref{eqmodifiedBessel}) then becomes the Airy equation:
\begin{equation}
\left(\partial_{z}^{2}-z\right)v(z)=0\,.
\end{equation}
This offers a more geometric picture of the non-analytic behavior of the modified Bessel functions. Applying the Laplace transform, one obtains:
\begin{equation}
v(z)={1\over 2\pi i}\int_{\gamma} e^{zw}e^{-{1\over 3}w^{3}}dw\,.
\end{equation}
The integral contour $\gamma$ on the $w$-plane has to be chosen such that the integrand vanishes at the boundary.
There are essentially three choices for them given by $\gamma_{i},i=1,2,3$, satisfying the homology relation
$\gamma_{1}+\gamma_{2}+\gamma_{3}\sim 0$ as depicted in Fig.~\ref{contours}.

\begin{figure}[h]
\centering
\includegraphics[width=0.3\textwidth]{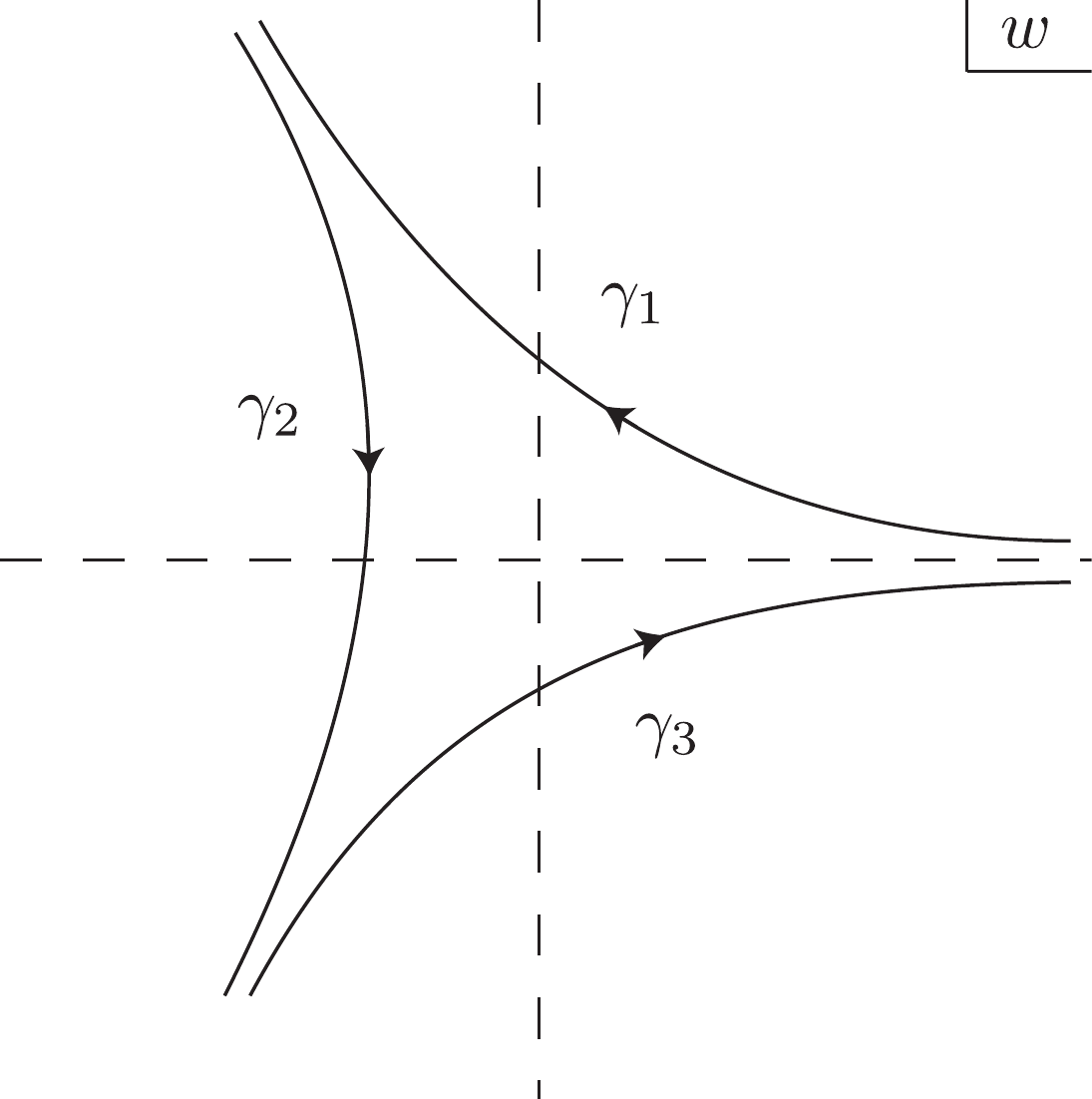} 
\caption{The contours $\gamma_1,\gamma_2$ and $\gamma_3$ on the  $w-$plane.}
\label{contours}
\end{figure}

The Airy function $\mathrm{Ai}(z)$ as a solution to this equation corresponds to the modified Bessel function $K_{1\over 3}$ and is given by:
\begin{equation}
\mathrm{Ai}(z)={1\over \pi \sqrt{3}}z^{1\over 2}K_{{1\over 3}}\left({2\over 3} z^{3\over 2}\right)
={1\over 2\pi i}\int_{-\gamma_{2}} e^{zw}e^{-{1\over 3}w^{3}}dw\,.
\end{equation}
It has exponential decay as $z\rightarrow \infty, \arg z=0$ and is oscillating as $z\rightarrow \infty, \arg z=\pi$. 
The other independent solution is usually taken to be the function
\begin{equation}
\mathrm{Bi}(z)=\sqrt{{z\over 3}} \left(     I_{{1\over 3}}\left({2\over 3} z^{3\over 2}\right)+I_{-{1\over 3}}\left({2\over 3} z^{3\over 2}\right)    \right)
={1\over 2\pi i}\int_{i\gamma_{2}+2i\gamma_{3}} e^{zw}e^{-{1\over 3}w^{3}}dw\,.
\end{equation}
It has exponential growth as $z\rightarrow \infty, \arg z=0$ and is oscillating as $z\rightarrow \infty, \arg z=\pi$. 
The same integral would have different asymptotic series expansions as the phase of $z$ changes, exhibiting Stokes phenomena as discussed in e.g., Refs.~\cite{Marino:2012zq,Witten:2010cx}.

The partition function can be written as:
\begin{equation}
\mathcal{Z}_{top,s}
=2^{1\over 3}e^{-{1\over 3\lambda_{s}^{2}}} \lambda_{s}^{-{1\over 3}}v\left(\left({1\over 2\lambda_{s}^2}\right)^{2\over 3}\right)\,.
\end{equation}
The integral over $-\gamma_{2}$ gives $\mathrm{Ai}(z)$ and thus the purely non-perturbative part of $\mathcal{Z}_{top,s}$, while the integrals over the other two cycles $\gamma_{1},\gamma_{3}$ are equivalent modulo the non-perturbative expressions and contribute to the perturbative part of $\mathcal{Z}_{top,s}$. Analytic continuation and exploring the non-perturbative content can thus be realized via moving the integral contour as discussed in many other contexts, such as in Ref.~\cite{Witten:2010cx}.


\section{Conclusion and discussion}\label{seccon}

In this work we determined certain terms of the topological string free energies to all orders in perturbation theory universally for any CY threefold. Using the polynomial formulation of the holomorphic anomaly equations we derived an Airy equation in a rescaled topological string coupling which has a solution encoding these terms. This solution admits a strong coupling expansion which should be the analog of an expansion considered in Ref.~\cite{Drukker:2010nc} for ABJM theory \cite{Aharony:2008ug}. It would be exciting to develop a similar interpretation which can be attached to any CY background. 

A second linearly independent solution of the Airy equation does not contribute to the perturbative expansion around $\lambda_s=0$ but gives contributions away from this value. This is a manifestation of a non-perturbative ambiguity attached to differential equations as it appears in many other contexts, see Ref.~\cite{Marino:2012zq}. This may offer geometric hints of the non-perturbative structure of topological string theory. Indeed the Airy function appears in the recent work \cite{Grassi:2014zfa} defining non-perturbative topological strings for some non-compact CY threefolds. 

Differential equations in the parameters of a theory often suggest an underlying geometric interpretation and perhaps a variation problem associated to it. Our results may provide the first steps in this direction. It is perhaps suggestive to think about the topological string coupling $\lambda$ as a complex variable, paralleling the discussion of analytic continuation of Chern-Simons theory \cite{Witten:2010cx}. Furthermore, the non-perturbative structure of topological strings was also addressed using methods of trans-series and resurgence of differential equations, see e.g., Ref.~\cite{Couso-Santamaria:2014iia} and references therein. These methods were applied to the master anomaly equation which is not a differential equation in the topological string coupling. Our results show that a differential equation in the coupling can be deduced, giving further motivation to pursue this line of research.

The limit described in the paper seems to select only the cubic graphs or equivalently the geometric genus zero contribution to the topological string partition function. The similarities to Refs.~\cite{Witten:1990hr, Kontsevich:1992ti,Aganagic:2003qj} suggest to further study an associated integrable hierarchy structure of the partition function. It is natural to expect, from examining the Feynman diagrams, that the lower degree terms in the polynomial generator should correspond to Riemann surfaces with higher geometric genera. In this sense $\mathcal{Z}_{top,s}$ is really the classical part with respect to the loop expansion in the parameter $\varepsilon$ assigned to the geometric genus, and the full partition function $\mathcal{Z}_{top}$ is its quantization. This parameter plays an analogous role as the equivariant parameter in localization. The interpretation of this parameter in the A-model would have the effect of reducing the equivariant localization to the stratum corresponding to the most degenerate configurations in the moduli space of curves. It would be furthermore interesting to understand a more direct physical meaning of the parameter $\varepsilon$. These questions will be addressed elsewhere.

\subsection*{Acknowledgements}
We would like to thank Peter Mayr, Ilarion Melnikov and especially Marcos Marino for very helpful discussions and comments on the manuscript. We would also like to thank Kevin Costello, Jaume Gomis, Daniel Jafferis, Hee-Cheol Kim, Peter Koroteev and Si Li for discussions and correspondence.  M.~A. is supported by NSF grant PHY-1306313. J.~Z.~is supported by the Perimeter Institute for Theoretical Physics. Research at Perimeter Institute is  supported  by  the  Government  of  Canada  through  Industry  Canada  and  by  the
Province of Ontario through the Ministry of Economic Development and Innovation.


\providecommand{\href}[2]{#2}\begingroup\raggedright\endgroup


\end{document}